\documentclass[sn-mathphys,Numbered]{sn-jnl}

\usepackage{threeparttable}

\usepackage{subfigure} 
\usepackage{ulem}
\usepackage{graphicx}%
\usepackage{multirow}%
\usepackage{amsmath,amssymb,amsfonts}%
\usepackage{amsthm}%
\usepackage{mathrsfs}%
\usepackage[title]{appendix}%
\usepackage{xcolor}%
\usepackage{textcomp}%
\usepackage{manyfoot}%
\usepackage{booktabs}%
\usepackage{algorithm}%
\usepackage{algorithmicx}%
\usepackage{algpseudocode}%
\usepackage{listings}%

\theoremstyle{thmstyleone}%
\newtheorem{thm}{Theorem}[section]
\newtheorem{proposition}[thm]{Proposition}%
\makeatletter

\newcommand{\Rmnum}[1]{\expandafter\@slowromancap\romannumeral #1@}
\makeatother
\usepackage{amsfonts}
\usepackage{txfonts}
\usepackage{paralist}
\usepackage{color}

\usepackage{amsmath,bm}
\usepackage{amstext}
\usepackage{amsfonts}
\usepackage{amssymb}
\usepackage{graphicx}
\usepackage{verbatim}
\usepackage{bbm}

\newtheorem{lem}{Lemma}

\theoremstyle{thmstyletwo}%
\newtheorem{example}{Example}%
\newtheorem{rmk}{Remark}%

\theoremstyle{thmstylethree}%
\usepackage{graphicx} 
\usepackage{subfigure} 
\usepackage{ulem}
\newtheorem{definition}{Definition}%

\begin{document}

\title[Article Title]{Simultaneous test of the mean vectors and covariance matrices for  high-dimensional data using RMT}


\author[1]{\fnm{Zhenzhen} \sur{Niu}} 

\author[2]{\fnm{Jianghao } \sur{Li}}

\author[3]{\fnm{Wenya} \sur{Luo}}


\author*[2]{\fnm{Zhidong} \sur{Bai}}\email{baizd@nenu.edu.cn}

\affil[1]{\orgdiv{ School of Mathematics and Statistics}, \orgname{Shandong Normal University}, \orgaddress{\street{ No.1 University Road,
			Science Park, Changqing District}, \city{Jinan}, \postcode{250358}, \country{China}}}

\affil[2]{\orgdiv{School of Mathematics and Statistics}, \orgname{Northeast Normal University}, \orgaddress{\street{5268 Peoples Road}, \city{Changchun}, \postcode{130024}, \country{China}}}

\affil[3]{\orgdiv{School of Data Sciences}, \orgname{ Zhejiang University of Finance and Economics}, \orgaddress{\street{Xueyuan Street, Xiasha Higher Education Park}, \city{Hangzhou}, \postcode{310018}, \country{China}}}


\abstract{
In this paper, we propose a new modified likelihood ratio test (LRT) for simultaneously testing mean vectors and covariance matrices of two-sample populations in high-dimensional settings. By employing tools from Random Matrix Theory (RMT), we derive the limiting null distribution of the modified LRT for generally distributed populations. Furthermore, we compare the proposed test with existing tests using simulation results, demonstrating that the modified LRT exhibits favorable properties in terms of both size and power.
}




\maketitle

\section{Introduction}\label{sec1}

Large-dimensional data sets are becoming increasingly prevalent, involving many scientific fields such as biology, medicine, finance and so on, as modern data collecting and processing technology advances. Traditional hypothesis testing methods in multivariate statistical analysis are no longer effective or perform poorly. It is a challenging problem to develop effective methods for statistical inference of high-dimensional data sets.  In recent years, the test for two sample mean vectors (see, e.g., \cite{bai1996effect}; \cite{chen2010two}; \cite{srivastava2008test}) or covariance matrices (see, e.g., \cite{bai2009corrections}; \cite{li2012two}; \cite{cai2013two}) under high-dimensional setting has been a very hot topic in the literature because of its important applications. We refer to \cite{hu2016review} and  \cite{huang2022overview} for recent comprehensive reviews on these topics. However, if we only test for mean vectors or covariance matrices, we may not effectively infer the differences between the two populations. Therefore, it makes sense to develop a new simultaneous test program for high-dimensional data. 

Assume that there are two independent $p$-dimensional populations
and let $\{\mathbf{x}_{i}^{(t)} \in \mathbb{R}^{p}, t=1,2, i=1,\ldots,N_{t} \}$ be independent and identically distributed
(i.i.d.) random sample vectors from the $t$-th population with mean vectors $\bm{\mu}_{t}=(\mu_{t1},\ldots,\mu_{tp})'$ and covariance matrices $\boldsymbol{\Sigma}_{t}$, $t=1, 2$,  respectively.
The simultaneous testing of the mean vectors and covariance matrices among two populations can be formulated as follows:
\begin{align}\label{H0}
\mathrm{H}_0: \bm{\mu}_{1}=\bm{\mu}_{2} ~\mbox{and}~ \boldsymbol{\Sigma}_{1}=\boldsymbol{\Sigma}_{2} \hspace{3ex} \text{v.s.} \hspace{3ex} \mathrm{H}_1: \bm{\mu}_{1} \neq \bm{\mu}_{2} ~\mbox{or} ~  \boldsymbol{\Sigma}_{1} \neq \boldsymbol{\Sigma}_{2}.
\end{align}

For testing hypotheses (\ref{H0}), there exist some conventional methods whose asymptotic properties are established in the
regime where the data dimension $p$ is fixed, and $n$ tends to infinity. See, e.g., when  $p<\min \left\{N_{1}-1, N_{2}-1\right\}$, the likelihood ratio test for $\mathrm{H}_{0}$,
\begin{align*}\label{LRT1}
\Lambda_{N} =\frac{\prod_{t=1}^{2}\left| \mathbf{A}_{t}^{\mathbf{x}} \right|^{N_{t} / 2} N^{p N / 2}}{\left|\left(N_{1} N_{2}\right) N^{-1}\left( \bar{\mathbf{x}}^{(1)}-\bar{\mathbf{x}}^{(2)} \right)\left( \bar{\mathbf{x}}^{(1)}-\bar{\mathbf{x}}^{(2)}\right)^{\prime}+\sum_{g=1}^{2} \mathbf{A}_{t}^{\mathbf{x}} \right|^{n / 2} \prod_{t=1}^{2} N_{t}^{p N_{t} / 2}},
\end{align*}
in \cite{anderson2003introduction},
where
\begin{align*}
\bar{\mathbf{x}}^{(t)}=\frac{1}{N_{t}}\sum_{i=1}^{N_{t}}\mathbf{x}_{i}^{(t)},\quad  \mathbf{A}_{t}^{\mathbf{x}}=\sum_{i=1}^{N_{t}}(\mathbf{x}_{i}^{(t)}-\bar{\mathbf{x}}^{(t)})(\mathbf{x}_{i}^{(t)}-\bar{\mathbf{x}}^{(t)})^{'}, t=1,2
\end{align*}
and $N=N_{1}+N_{2}$.
However, the likelihood ratio is not well defined when $p>\min \left\{N_{1}-1, N_{2}-1\right\}$. In the sequel, in order to make up the deficiency of the traditional test and solve the large dimension problem, new methods have been proposed by researchers.

Some related works on the corrected likelihood ratio test from different perspectives can be found in \cite{jiang2013central} and \cite{jiang2015likelihood}. When the population was the multivariate normal distribution, \cite{jiang2013central} and \cite{jiang2015likelihood} studied LRT under the different assumption that the data dimension grew to infinity but was smaller than simple size, respectively. As a complement of \cite{jiang2015likelihood}, \cite{chen2018study} obtained the asymptotic distribution of the log-likelihood ratio test statistic under alternative hypotheses that were not local ones. 
For more extensions and issues on this line, we refer the reader to \cite{lim2010likelihood}; \cite{bai2023moderate}.

However, a few works investigate simultaneous test of high-dimensional mean vectors and covariance matrices in different settings with general populations. The problem of one-sample simultaneous test procedure based on the quadratic loss for covariance matrix estimation is proposed by \cite{liu2017simultaneous}.
 More recently, \cite{niu2019lr} further focused on the classical likelihood ratio
 test with high-dimensional non-Gaussian data. \cite{cao2021simultaneous} and \cite{hyodo2018simultaneous} considered a weighted sum of one test statistic related to the $L^{2}$-norm-based test for mean vectors and another test statistic related to the $L^{2}$-norm-based for covariance matrices in the context of the one-sample simultaneous test procedure and the two-sample simultaneous test procedure, respectively. However, their focus was limited to the structure or eigenvalues of the population covariance matrices.
 The test statistics of existing
 works are generally structured by weighting one proposed by the high-dimensional mean vectors test and another proposed by the high-dimensional covariance matrices test; see, e.g., \cite{liu2017simultaneous};  \cite{cao2021simultaneous};  \cite{hyodo2018simultaneous}. The weighting combination of such statistics may cause some loss of power; for more details, we may refer to the numerical simulations.
 In contrast, although the LRT statistic requires the dimensions to be smaller than the sample sizes, it makes no assumptions about the structure of the population covariance matrices. Inspired by the above discussions, we propose a novel modified likelihood ratio test for simultaneous testing of mean vectors and covariance matrices of two-sample populations. In order to solve the possible difficulties in the proof for non-Gaussian populations, we make full use of the advanced tools of RMT to accommodate the high-dimensional data.
The rest of the paper is organized as follows.
In Section ~\ref{sec2}, we present the main results of the new modified LRT for simultaneous testing of equalities of mean vectors and  covariance matrices.
Section ~\ref{sec3} shows some simulation results on the empirical sizes and the empirical powers of our test, including a comparison to other criteria under nonnormally  distributed data. Section ~\ref{sec5} contains some conclusions and discussions. The proof of the main result is  provided  in Appendix ~\ref{sec4}.

\section{Main Results}\label{sec2}

In this section, we present the asymptotic distribution of the modified LRT.
For convenience of exposition, we first introduce some notations.
In the following, the notation $\xrightarrow{D}$ and $\xrightarrow{p}$ denote  convergence in distribution and convergence in probability,  respectively.
$\delta_{(\cdot)}$ and $\log(\cdot)$  represent the  indicator function and logarithm function,  respectively.
Let
\begin{equation*}
\begin{split}
&n=n_{1}+n_{2},   \quad r_{n}=\frac{p}{n}, \quad y_{1}:=y_{n_{1}}=\frac{p}{n_{1}}, \quad y_{2}:=y_{n_{2}}=\frac{p}{n_{2}}, \quad h:=h_{n}=\sqrt{y_{1}+y_{2}-y_{1}y_{2}},\\
&c_{1}:=\frac{n_{1}}{n_{1}+n_{2}}=\frac{y_{2}}{y_{1}+y_{2}}, \quad  c_{2}:=\frac{n_{2}}{n_{1}+n_{2}}=\frac{y_{1}}{y_{1}+y_{2}},\\
&l(y_{1},y_{2})=\log(h^{\frac{2c_{1}h^{2}}{y_{1}y_{2}}})\delta_{y_{1}>1}-\log(y_{1}^{\frac{c_{1}(1+y_{2})}{y_{2}}}y_{2}^{\frac{c_{1}(1-y_{1})}{y_{1}}})\delta_{y_{1}>1},\\
&u(y_{1},y_{2})=\log(\frac{y_{1}^{c_{1}}}{h^{c_{1}}})\delta_{y_{1}>1}, \quad v(y_{1},y_{2})=\log(\frac{y_{1}^{2c_{1}}}{h^{2c_{1}(c_{1}+2c_{2})}})\delta_{y_{1}>1},\\
&\Psi(y_{1},y_{2})=c_{2}y_{1}^{2}[y_{2}^{4}\delta_{y_{2}<1}+h^{2}(2y_{2}^{2}-h^{2})\delta_{y_{2}>1}]
-c_{1}y_{2}^{2}[y_{1}^{3}(y_{1}+2y_{2})\delta_{y_{1}<1}+h^{2}(y_{1}+y_{2}+y_{1}y_{2})\delta_{y_{1}>1}].
\end{split}
\end{equation*}
In what follows, we start to consider the hypothesis (\ref{H0}).
Suppose that the observations    $\{\mathbf{x}_{1}^{(t)},\ldots,\mathbf{x}_{N_{t}}^{(t)}\}$ are $p$-dimensional i.i.d. random sample vectors from the $t$-th population for $t=1,2$,
which have the mean vector $\bm{\mu}_{t}=(\mu_{t1},\ldots,\mu_{tp})'$ and covariance matrix $\boldsymbol{\Sigma}_{t}$.
Now we assume that the observation $\mathbf{x}_{i}^{(t)}$ follows  the general multivariate model:
\begin{eqnarray}\label{model}
\mathbf{x}_{i}^{(t)}=\boldsymbol{\Sigma}_{t}^{\frac{1}{2}}\mathbf{z}_{i}^{(t)}+\bm{\mu}_{t},
\end{eqnarray}
for $i=1,\ldots,N_{t}$, $t=1$ and $2$,   where $\mathbf{z}_{i}^{(t)}=(z_{i1}^{(t)},\ldots,z_{ip}^{(t)})'$  and $\{z_{ij}^{(t)}, i\leq N_{t}, j\leq p\}$ are i.i.d.  real random variables.
	\cite{anderson2003introduction}, see e.g., Section 10.3, suggested the use of a modified LRT statistic $\Lambda_{n}^{*}$,
which uses  $n_{t}=N_{t}-1$ instead of sample size $N_{t}$,  $t=1,2$ and $N$ is replaced by $n=n_{1}+n_{2}=N-2$ in $\Lambda_{N}$.

Then under the null hypothesis  \eqref{H0} and linear transformation model \eqref{model}, we can easily rewrite the modified LRT statistic  $\Lambda_{n}^{*}$ as
\begin{equation*} \label{LRT2}
\begin{split}
\Lambda_{n}^{*}=\frac{|N_{1}^{-1}n_{1}\mathbf{S}_{1}^{\mathbf{z}}|^{\frac{N_{1}}{2}} \cdot |N_{2}^{-1}n_{2}\mathbf{S}_{2}^{\mathbf{z}}|^{\frac{N_{2}}{2}}}{|N^{-1}(n_{1}\mathbf{S}_{1}^{\mathbf{z}}+n_{2}\mathbf{S}_{2}^{\mathbf{z}})|^{\frac{N}{2}}}
&\cdot(\frac{|n_{1}\mathbf{S}_{1}^{\mathbf{z}}+n_{2}\mathbf{S}_{2}^{\mathbf{z}}|}{|(n_{1}\mathbf{S}_{1}^{\mathbf{z}}+n_{2}\mathbf{S}_{2}^{\mathbf{z}})+N_{1}N_{2}N^{-1}(\bar{\mathbf{z}}^{(1)}-\bar{\mathbf{z}}^{(2)})(\bar{\mathbf{z}}^{(1)}-\bar{\mathbf{z}}^{(2)})'|})^{\frac{N}{2}},
\end{split}
\end{equation*}
where  $\bar{\mathbf{z}}^{(t)}=N_{t}^{-1}\sum_{i=1}^{N_{t}}\mathbf{z}_{i}^{(t)}$ and
$\mathbf{S}_{t}^{\mathbf{z}}=n_{t}^{-1}\sum_{i=1}^{N_{t}}(\mathbf{z}_{i}^{(t)}-\bar{\mathbf{z}}^{(t)})(\mathbf{z}_{i}^{(t)}-\bar{\mathbf{z}}^{(t)})^{'}$, $t=1, 2$.
In what follows, we remove the superscripts of $\mathbf{S}_{1}^{\mathbf{z}}$ and $\mathbf{S}_{2}^{\mathbf{z}}$ and  we denote $\mathbf{S}_{1}:=\mathbf{S}_{1}^{\mathbf{z}}$ and  $\mathbf{S}_{2}:=\mathbf{S}_{2}^{\mathbf{z}}$ for brevity.

However, \cite{perlman1980unbiasedness}  has proved that this is the likelihood ratio test itself, not the modified test,
which is unbiased to test  \eqref{H0}.
 Thus, in this paper we reconsider the  statistic  $\Lambda_{n}^{*}$,

\begin{equation*}
\begin{split}
\Lambda_{n}^{*}=\frac{|\mathbf{S}_{1}|^{\frac{n_{1}}{2}} \cdot |\mathbf{S}_{2}|^{\frac{n_{2}}{2}}}{|n^{-1}(n_{1}\mathbf{S}_{1}+n_{2}\mathbf{S}_{2})|^{\frac{n}{2}}}
\cdot(\frac{|n_{1}\mathbf{S}_{1}+n_{2}\mathbf{S}_{2}|}{|(n_{1}\mathbf{S}_{1}+n_{2}\mathbf{S}_{2})+n_{1}n_{2}n^{-1}
	(\bar{\mathbf{z}}^{(1)}-\bar{\mathbf{z}}^{(2)})(\bar{\mathbf{z}}^{(1)}-\bar{\mathbf{z}}^{(2)})'|})^{\frac{n}{2}}.
\end{split}
\end{equation*}
Then after a simple calculation, we can obtain
\begin{eqnarray*}
	L=\frac{2}{n}\log\Lambda_{n}^{*}
	=c_{1}\cdot \log|c_{1}^{-1}\mathbf{B}_{n}|+c_{2}\cdot \log|c_{2}^{-1}(\mathbf{I}_{p}-\mathbf{B}_{n})|-\log(1+\mathrm{T}_{n}),
\end{eqnarray*}
where
$\mathbf{B}_{n}=n_{1}\mathbf{S}_{1}(n_{1}\mathbf{S}_{1}+n_{2}\mathbf{S}_{2})^{-1}$,
$\mathrm{T}_{n}=n_{1}n_{2}n^{-1}(\bar{\mathbf{z}}^{(1)}-\bar{\mathbf{z}}^{(2)})^{'}(n_{1}\mathbf{S}_{1}+n_{2}\mathbf{S}_{2})^{-1}(\bar{\mathbf{z}}^{(1)}-\bar{\mathbf{z}}^{(2)})$
and $\mathbf{I}_{p}$ is the $p \times p$ identity matrix.

Note that when $p>n_{1}$ or $p>n_{2}$,  $L$ is undefined. Morever, we know from previous works (see \cite{bai2015convergence} ) that if the fourth moment of $z_{ij}^{(t)}$ exists, matrix $\mathbf{B}_{n}$ almost certainly contains $p-n_{1}$ zero eigenvalues and $p-n_{2}$ one eigenvalues for the conditions $p>\min \left\{N_{1}-1, N_{2}-1\right\}$.
Therefore, we redefine $L$  by restricting the eigenvalues of $\mathbf{B}_{n}$ between zero and one, that is
\begin{eqnarray}\label{LRT1}
L= \sum_{\lambda_{i}^{\mathbf{B}_{n}}\in(0,1)}[c_{1}\log\lambda_{i}^{\mathbf{B}_{n}} + c_{2} \log(1-\lambda_{i}^{\mathbf{B}_{n}})]-\log(1+\mathrm{T}_{n}).
\end{eqnarray}
where $\lambda_{i}^{\mathbf{B}_{n}}$ denotes the $i$-th smallest eigenvalues of  $\mathbf{B}_{n}$.

Next, we impose the following two assumptions, which are frequently utilized in random matrix theory, to analyze the asymptotic behaviors of the considered statistics throughout the paper.
	\begin{itemize}
	\item \textbf{Assumption A:} The random vectors $\mathbf{z}_{i}^{(t)}=(z_{i1}^{(t)},\ldots,z_{ip}^{(t)})'$ satisfy the model (\ref{H0}) with common moments $\mathbb{E}z_{11}^{(1)}=\mathbb{E}z_{11}^{(2)}=0$, $\mathbb{E}(z_{11}^{(1)})^{2}=\mathbb{E}(z_{11}^{(2)})^{2}=1$, and  $\mathbb{E}(z_{11}^{(1)})^{4}=\beta_{1}+3 <\infty$ and  $\mathbb{E}(z_{11}^{(2)})^{4}=\beta_{2}+3 <\infty$;
	\item \textbf{Assumption B:} $y_{1}\neq 1$, $y_{2}\neq 1$ and $r_{n} <1$, $\lim y_{1} \notin \{0, 1\} $, $\lim y_{2}\notin \{0, 1\} $
	and $r=\lim r_{n} \in (0, 1)$ as
	$\min\{ p, n_{1}, n_{2}\} \rightarrow \infty$;
\end{itemize}
\begin{rmk}
 In Assumption B, we require the data dimension $p$ to be less than $n_{1} + n_{2}$  so that
 the matrix $n_{1}\mathbf{S}_{1}+n_{2}\mathbf{S}_{2}$ is invertible, but allow that the data dimension $p$ to be greater or less than the sample size $n_{t}, t=1,2$.
\end{rmk}

Subsequently,  the following theorem establishes the joint limiting null distribution of $L$ defined in \eqref{LRT1}.
\begin{thm}\label{thm}
	Under the assumptions A and B , and the null hypothesis $\mathrm{H}_{0}$ in (\ref{H0}), we have
	\begin{align*}
		\frac{L-p \cdot l_{n}-\mu_{n}-\log(1-r_{n})}{\nu_{n}} \xrightarrow{D} N(0,1),
	\end{align*}
	where
	\begin{equation*}\label{eq:1}
	\begin{split}
	l_{n}=&\log(\frac{y_{1}^{c_{2}}y_{2}^{c_{1}}h^{\frac{2h^{2}}{y_{1}y_{2}}}}{(y_{1}+y_{2})^{\frac{(y_{1}+y_{2})}{y_{1}y_{2}}}|1-y_{1}|^{\frac{c_{1}|1-y_{1}|}{y_{1}}}|1-y_{2}|^{\frac{c_{2}|1-y_{2}|}{y_{2}}}})-l(y_{1},y_{2})-l(y_{2},y_{1}),\\
	\mu_{n}=&\log[\frac{(y_{1}+y_{2})^{\frac{1}{2}}|1-y_{1}|^{\frac{c_{1}}{2}}|1-y_{2}|^{\frac{c_{2}}{2}}}{h}]-u(y_{1},y_{2})-u(y_{2},y_{1})\\
	&+\frac{\beta_{1}\Psi(y_{1},y_{2})}{2y_{1}y_{2}^{2}(y_{1}+y_{2})^{2}}+\frac{\beta_{2}\Psi(y_{2},y_{1})}{2y_{2}y_{1}^{2}(y_{1}+y_{2})^{2}},\\
	\nu_{n}^{2}=&\log\frac{h^{4}}{|1-y_{1}|^{2c_{1}^{2}}|1-y_{2}|^{2c_{2}^{2}}(y_{1}+y_{2})^{2}}\\
	&+2[v(y_{1},y_{2})+v(y_{2},y_{1})+\log(h^{4c_{1}c_{2}})\delta_{y_{1}>1}\delta_{y_{2}>1}]\\
	&+\frac{(y_{1}\beta_{1}+y_{2}\beta_{2})}{y_{1}^{2}y_{2}^{2}(y_{1}+y_{2})^{2}}[(y_{1}-1)y_{2}^{2}\delta_{y_{1}>1}-(y_{2}-1)y_{1}^{2}\delta_{y_{2}>1}]^{2}.
	\end{split}
	\end{equation*}
\end{thm}

\begin{rmk}
	Note that if $y_{1}$ or  $y_{2}$  is close to  1,  the variance $\nu_{n}$ will increase rapidly  and the LRT will become unstable.
\end{rmk}

	Let $\alpha$ be the significance level.
	According to Theorem \ref{thm},
	the rejection region of the test problem \eqref{H0} is denoted  by
	\begin{eqnarray*}
		&& \{(\mathbf{x}_{1},\ldots,\mathbf{x}_{n}) : L< -\nu_{n}z_{\alpha/2}+p \cdot l_{n}+\mu_{n}+\log(1-r_{n}) \quad or \\
		&&  L> \nu_{n}z_{\alpha/2}+p \cdot l_{n}+\mu_{n}+\log(1-r_{n})   \},
	\end{eqnarray*}
	where $z_{\alpha}$ is the upper $\alpha$ quantile of the standard normal distribution,
	$\nu_{n}$, $\mu_{n}$ and $l_{n}$ are defined in  Theorem \ref{thm}.

\begin{rmk}
		Note that  the fourth moments are unknown in practical applications, and therefore, the estimates  of the fourth moments are  necessary.  \cite{zhang2019invariant} obtained their consistent estimators by using the method of moments and RMT, more details may be found in the Theorem 2.7 in \cite{zhang2019invariant}

\end{rmk}

The  proof of  the  Theorem \ref{thm} will be postponed in Appendix \ref{sec4}.

\section{Simulation study}\label{sec3}
In this section, we conduct simulation studies to illustrate the performance of the proposed modified LRT (referred to as the ML test in the following context), compared to the tests proposed by \cite{hyodo2018simultaneous} (referred to as the HN test). 

Assume that the random samples are generated from the following model:
\begin{align*}
\mathbf{x}_{i}^{(t)}=\boldsymbol{\Sigma}_{t}^{\frac{1}{2}}\mathbf{z}_{i}^{(t)}+\bm{\mu}_{t}, \quad i=1,\ldots,N_{t}, \hspace{1ex}  t=1 ~\mbox{and}~  2,
\end{align*}
where $\mathbf{z}_{i}^{(t)}=(z_{i1}^{(t)},\ldots,z_{ip}^{(t)})'$  and $\{z_{ij}^{(t)}, j=1,\ldots,p\}$ are i.i.d. real random variables from Gamma distribution $Gamma(4,2)-2$.
Note that in the following simulations, we did not utilize the estimators $\hat{\beta}_{i}$, $i=1,2$,  and we always assume that the fourth moments of the $\{z_{ij}^{(t)}\}_{j=1}^{p}$ are known. \cite{zhang2019invariant} has provided explanations and demonstrated that the performance of the estimators is remarkable based on the numerical results. For further details, we direct readers to \cite{zhang2019invariant}.

For the mean vectors  $\bm{\mu}_{1}$ and $\bm{\mu}_{2}$ and the covariance matrices  $\boldsymbol{\Sigma}_{1}$ and $\boldsymbol{\Sigma}_{2}$, we consider two scenarios with respect to the random samples $\{\mathbf{x}_{i}^{(1)}\}_{i=1}^{N_{1}}$ and $\{\mathbf{x}_{i}^{(2)}\}_{i=1}^{N_{2}}$ :

$\bullet$ Model \Rmnum{1}: $\bm{\mu}_{1}=\bm{\mu}_{2}=\bm0_{p}$, $\boldsymbol{\Sigma}_{2}=diag(p^{2},1,\ldots,1)$, $\boldsymbol{\Sigma}_{1}=(1+a/n_{1})\boldsymbol{\Sigma}_{2}$,
where the number of $p^{2}$ is equal to 1 and $a$ is a constant, and $\bm0_{p}$  denotes  the $p$-dimensional zero vector.

$\bullet$ Model \Rmnum{2}: $\bm{\mu}_{1}=(1,p,\ldots,p)'$, $\bm{\mu}_{2}=(1,p+1,\ldots,p+1)'$, $\boldsymbol{\Sigma}_{1}=\boldsymbol{\Sigma}_{2}=diag(p^{2},1,\ldots,1)$, where the number of $p$ and $p+1$ is equal to $(p-1)$, respectively.

In each case, we perform 10,000 independent replications to estimate the empirical sizes and the empirical powers of the proposed ML test and the HN test based on different values of $(n_{1}, n_{2}, p)$, respectively, and the nominal significance level of the tests was  $\alpha=0.05$.

Table \ref{tab1} displays the empirical sizes  of the proposed ML test and the HN test under Model \Rmnum{1} with $a=0$. In addition, Table \ref{tab2} and Table \ref{tab3} present the empirical powers of the HN test and the proposed ML test for the two alternative settings, respectively.
 We observe from Table \ref{tab1} that under model \Rmnum{1} of Gamma distribution, the empirical sizes of  the proposed ML test  are closed to the significance level 0.05  when $n_{1}, n_{2}, p$  increase. In contrast, the empirical sizes of the HN test are slightly higher than ours did. This reflects that the null distribution of the test statistic of the HN test can not be approximated its asymptotic distribution well in this case. The power results in table \ref{tab2} showed the proposed ML test and the HN test had similar empirical power and and were less affected by the increased dimensionality when $y_1 <1, y_2>1$ or $y_1>1, y_2 <1$, but the proposed ML test had quite good power compared with the HN test  when $y_1 >1, y_2>1$ or $y_1<1, y_2 <1$ . The powers under the second model reported in Table \ref{tab3} increased much faster than those under the first model reported in Table \ref{tab2} as the sample size and the dimension increased, which indicated that the proposed ML test is more sensitive than the HN test in our setting.

 Fig.1-4 illustrate that the probability density curve of the standard normal distribution N(0,1) is consistent with  the histogram of the proposed ML test as  $p$  becomes larger.  The red curve represents the standard normal distribution density curve. Figures 5-8 show the divergence of powers between the two test statistics, with the parameter $a$ increasing, showing the trend from the empirical sizes to the empirical powers under Model \Rmnum{1}.




\begin{table}[h]
\caption{Empirical sizes of the tests  the proposed ML test and the HN test, based on 10,000 replications with real Gamma data under Model \Rmnum{1} with $a=0$.}\label{tab1}
\begin{tabular*}{\textwidth}{@{\extracolsep\fill}lcccccc}
\toprule%
$y_{1}>1, y_{2}>1$ &\multirow{2}{*}{$(25, 35, 40)$} &\multirow{2}{*}{$(50, 70, 80)$} &\multirow{2}{*}{$(100, 140, 160)$} &\multirow{2}{*}{$(200, 280, 320)$} \\
$(n_{1}, n_{2}, p)$   \\
\midrule
$\text{ML}$    &0.0639 &0.0566 &0.0559 &0.0496       \\
$\text{HN}$     &0.1335 &0.1431 &0.1384 &0.1474     \\
\botrule
\toprule%
$y_{1}>1, y_{2}<1$ &\multirow{2}{*}{$(25, 35, 30)$} &\multirow{2}{*}{$(50, 70, 60)$} &\multirow{2}{*}{$(100, 140, 120)$} &\multirow{2}{*}{$(200, 280, 240)$} \\
$(n_{1}, n_{2}, p)$     \\
\midrule
$\text{ML}$             &0.0572  &0.0573 &0.0533  &0.0526     \\
$\text{HN}$              &0.1375  &0.146  &0.1384  &0.1461    \\
\botrule
\toprule%
$y_{1}<1, y_{2}>1$ &\multirow{2}{*}{$(35, 25, 30)$} &\multirow{2}{*}{$(70, 50, 60)$} &\multirow{2}{*}{$(140, 100, 120)$} &\multirow{2}{*}{$(280, 200, 240)$} \\
$(n_{1}, n_{2}, p)$     \\
\midrule
$\text{ML}$                &0.0694   &0.063    &0.0582   &0.0498     \\
$\text{HN}$               &0.14     &0.1383   &0.1407   &0.1399     \\
\botrule
\toprule%
$y_{1}<1, y_{2}<1$ &\multirow{2}{*}{$(25, 35, 20)$} &\multirow{2}{*}{$(50, 70, 40)$} &\multirow{2}{*}{$(100, 140, 80)$} &\multirow{2}{*}{$(200, 280, 160)$} \\
$(n_{1}, n_{2}, p)$     \\
\midrule
$\text{ML}$             &0.062  &0.0549 &0.0539  &0.0543     \\
$\text{HN}$             &0.1313 &0.138  &0.1509  &0.1501    \\
\botrule
\end{tabular*}
\end{table}

\begin{table}[h]
\caption{Empirical powers of the tests the proposed ML test and the HN test, based on 10,000 replications with real Gamma data under Model \Rmnum{1}}.\label{tab2}
\begin{tabular*}{\textwidth}{@{\extracolsep\fill}lccccccccc}
\toprule%
   & \multirow{2}{*}{$(n_{1}, n_{2}, p)$} & \multicolumn{2}{@{}c@{}}{$a=5$} & \multicolumn{2}{@{}c@{}}{$a=10$} & \multicolumn{2}{@{}c@{}}{$a=15$}& \multicolumn{2}{@{}c@{}}{$a=20$}  \\
    \cmidrule{3-4}\cmidrule{5-6}\cmidrule{7-8}\cmidrule{9-10}%
&  &$\text{ML}$ &$\text{HN}$ &$\text{ML}$ &$\text{HN}$ &$\text{ML}$ &$\text{HN}$ &$\text{ML}$ &$\text{HN}$  \\
\midrule
$y_{1}>1$          &(25,35,40)        &0.2271    &0.1405     &0.4884    &0.1769    &0.752     &0.2478   &0.9061    &0.3066    \\
$y_{2}>1$          & (50,70,80)       &0.2153    &0.145      &0.5129    &0.1707    &0.7863    &0.2147   &0.9416    &0.2623   \\
                   &(100,140,160)     &0.2202    &0.1423     &0.5168    &0.161     &0.7954    &0.1822   &0.9514    &0.2065    \\
                   &(200,280,320)     &0.2139    &0.1491     &0.5232    &0.1586    &0.8096    &0.1759   &0.9548    &0.1916   \\
\botrule
\toprule%
   & \multirow{2}{*}{$(n_{1}, n_{2}, p)$} & \multicolumn{2}{@{}c@{}}{$a=5$} & \multicolumn{2}{@{}c@{}}{$a=10$} & \multicolumn{2}{@{}c@{}}{$a=15$}& \multicolumn{2}{@{}c@{}}{$a=20$}  \\
    \cmidrule{3-4}\cmidrule{5-6}\cmidrule{7-8}\cmidrule{9-10}%
&  &$\text{ML}$ &$\text{HN}$ &$\text{ML}$ &$\text{HN}$ &$\text{ML}$ &$\text{HN}$ &$\text{ML}$ &$\text{HN}$  \\
\midrule
$y_{1}>1$       &(25,35,30)      &0.1307     &0.1407    &0.2615    &0.1845    &0.4316    &0.2358     &0.6376    &0.3104 \\
$y_{2}<1$       &(50,70,60)      &0.1266     &0.1454    &0.2607    &0.1742    &0.4455    &0.2102     &0.6425    &0.2619\\
                &(100,140,120)   &0.1229     &0.14      &0.2454    &0.1559    &0.4358    &0.1859     &0.6292    &0.2173 \\
                &(200,280,240)   &0.1255     &0.1447    &0.2467    &0.1555    &0.4287    &0.1709     &0.6146    &0.1804  \\
\botrule
\toprule%
   & \multirow{2}{*}{$(n_{1}, n_{2}, p)$} & \multicolumn{2}{@{}c@{}}{$a=5$} & \multicolumn{2}{@{}c@{}}{$a=10$} & \multicolumn{2}{@{}c@{}}{$a=15$}& \multicolumn{2}{@{}c@{}}{$a=20$}  \\
    \cmidrule{3-4}\cmidrule{5-6}\cmidrule{7-8}\cmidrule{9-10}%
&  &$\text{ML}$ &$\text{HN}$ &$\text{ML}$ &$\text{HN}$ &$\text{ML}$ &$\text{HN}$ &$\text{ML}$ &$\text{HN}$  \\
\midrule
$y_{1}<1$       &(35,25,30)      &0.1124  &0.1538  &0.1463  &0.2044   &0.1843    &0.2531    &0.2121  &0.3028    \\
$y_{2}>1$       &(70,50,60)      &0.1018  &0.1478  &0.1557  &0.1856   &0.2116    &0.2113    &0.274   &0.2437    \\
                &(140,100,120)   &0.1017  &0.1554  &0.1575  &0.1645   &0.2372    &0.1841    &0.3074  &0.1999    \\
                &(280,200,240)   &0.102   &0.1518  &0.1587  &0.1605   &0.2419    &0.1679    &0.3519  &0.1695    \\
\botrule
\toprule%
   & \multirow{2}{*}{$(n_{1}, n_{2}, p)$} & \multicolumn{2}{@{}c@{}}{$a=20$} & \multicolumn{2}{@{}c@{}}{$a=40$} & \multicolumn{2}{@{}c@{}}{$a=60$}& \multicolumn{2}{@{}c@{}}{$a=80$}  \\
    \cmidrule{3-4}\cmidrule{5-6}\cmidrule{7-8}\cmidrule{9-10}%
&  &$\text{ML}$ &$\text{HN}$ &$\text{ML}$ &$\text{HN}$ &$\text{ML}$ &$\text{HN}$ &$\text{ML}$ &$\text{HN}$  \\
\midrule
$y_{1}<1$       &(25,35,20)      &0.2594     &0.3225    &0.7719    &0.6086    &0.9825     &0.8227     &0.9998    &0.9198  \\
$y_{2}<1$       &(50,70,40)      &0.1628     &0.2562    &0.6072    &0.5245    &0.956      &0.7605     &0.9986    &0.8929\\
                &(100,140,80)    &0.1115     &0.221     &0.3607    &0.4123    &0.7895     &0.6354     &0.9862    &0.8  \\
                &(200,280,160)   &0.0795     &0.1867    &0.1903    &0.3043    &0.4757     &0.476      &0.8164    &0.652\\
\botrule
\end{tabular*}
\end{table}

\begin{table}[h]
\caption{Empirical powers of the tests the proposed ML test and the HN test, based on 10,000 replications with real Gamma data under Model \Rmnum{2}}.\label{tab3}
\begin{tabular*}{\textwidth}{@{\extracolsep\fill}lcccccc}
\toprule%
$y_{1}>1, y_{2}>1$ &\multirow{2}{*}{$(25, 35, 40)$} &\multirow{2}{*}{$(50, 70, 80)$} &\multirow{2}{*}{$(100, 140, 160)$} &\multirow{2}{*}{$(200, 280, 320)$} \\
$(n_{1}, n_{2}, p)$   \\
\midrule
$\text{ML}$      &0.8896 &0.984  &0.9985 &0.9999       \\
$\text{HN}$     &0.1542 &0.1525 &0.1509 &0.1558     \\
\botrule
\toprule%
$y_{1}>1, y_{2}<1$ &\multirow{2}{*}{$(25, 35, 30)$} &\multirow{2}{*}{$(50, 70, 60)$} &\multirow{2}{*}{$(100, 140, 120)$} &\multirow{2}{*}{$(200, 280, 240)$} \\
$(n_{1}, n_{2}, p)$     \\
\midrule
$\text{ML}$               &0.7217 &0.9197  &0.9868  &0.9987    \\
$\text{HN}$              &0.165  &0.1692  &0.1605  &0.1626   \\
\botrule
\toprule%
$y_{1}<1, y_{2}>1$ &\multirow{2}{*}{$(35, 25, 30)$} &\multirow{2}{*}{$(70, 50, 60)$} &\multirow{2}{*}{$(140, 100, 120)$} &\multirow{2}{*}{$(280, 200, 240)$} \\
$(n_{1}, n_{2}, p)$     \\
\midrule
$\text{ML}$                &0.7168    &0.9146    &0.985    &0.999     \\
$\text{HN}$               &0.1545    &0.1588    &0.1687   &0.1634     \\
\botrule
\toprule%
$y_{1}<1, y_{2}<1$ &\multirow{2}{*}{$(25, 35, 20)$} &\multirow{2}{*}{$(50, 70, 40)$} &\multirow{2}{*}{$(100, 140, 80)$} &\multirow{2}{*}{$(200, 280, 160)$} \\
$(n_{1}, n_{2}, p)$     \\
\midrule
$\text{ML}$             &0.8778  &0.9897 &0.9998  &1     \\
$\text{HN}$            &0.1757  &0.1736 &0.1716  &0.1684    \\
\botrule
\end{tabular*}
\end{table}

\begin{figure}[htbp]
	\centering
	\subfigure{\includegraphics[height=3cm,width=3cm]{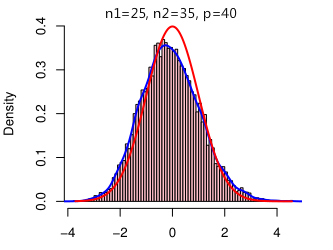}}
	\subfigure{\includegraphics[height=3cm,width=3cm]{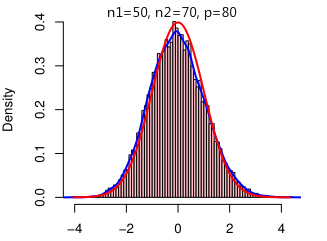}}
	\subfigure{\includegraphics[height=3cm,width=3cm]{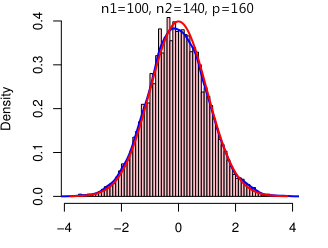}}
	\subfigure{\includegraphics[height=3cm,width=3cm]{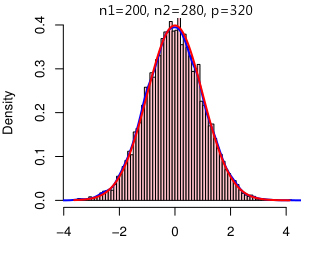}}
	\caption{
   The pictures show that for different $n_{1}, n_{2}, p$, that is when $y_{1}>1, y_{2}>1$, the probability density curve of the standard normal distribution N(0,1) is consistent with  the histogram of the proposed ML test as  $p$  becomes larger.  }
	\label{fig1}
\end{figure}

\begin{figure}[htbp]
	\centering
	\subfigure{\includegraphics[height=3cm,width=3cm]{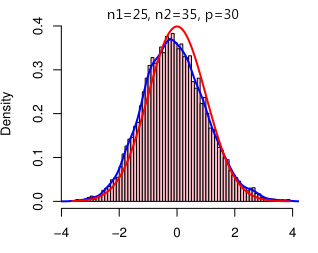}}
	\subfigure{\includegraphics[height=3cm,width=3cm]{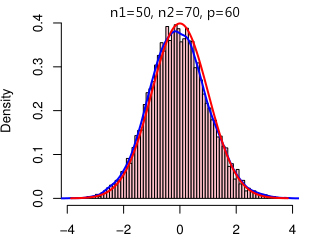}}
	\subfigure{\includegraphics[height=3cm,width=3cm]{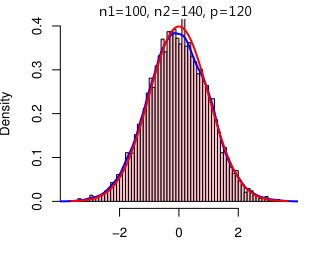}}
	\subfigure{\includegraphics[height=3cm,width=3cm]{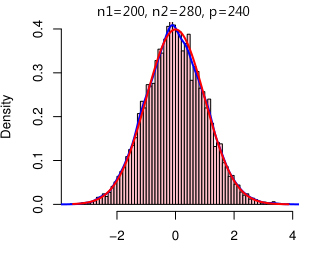}}
	\caption{
        The pictures show that for different $n_{1}, n_{2}, p$, that is when $y_{1}>1, y_{2}<1$, the probability density curve of the standard normal distribution N(0,1) is consistent with  the histogram of the proposed ML test as  $p$  becomes larger.  }
	\label{fig2}
\end{figure}

\begin{figure}[htbp]
	\centering
	\subfigure{\includegraphics[height=3cm,width=3cm]{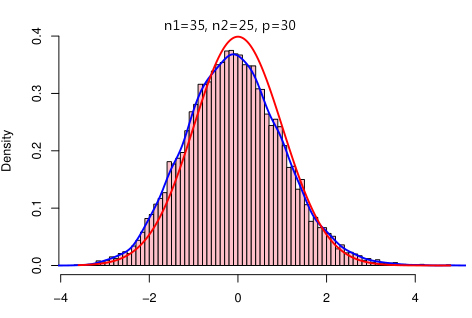}}
	\subfigure{\includegraphics[height=3cm,width=3cm]{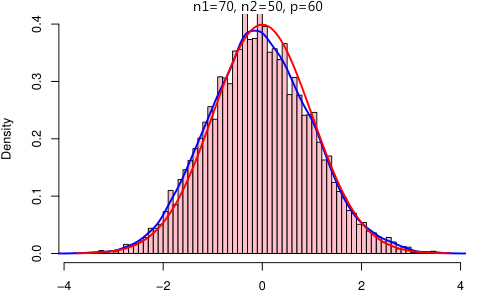}}
	\subfigure{\includegraphics[height=3cm,width=3cm]{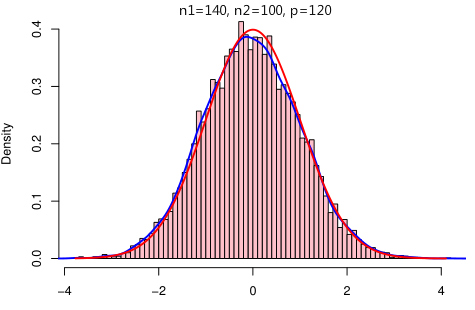}}
	\subfigure{\includegraphics[height=3cm,width=3cm]{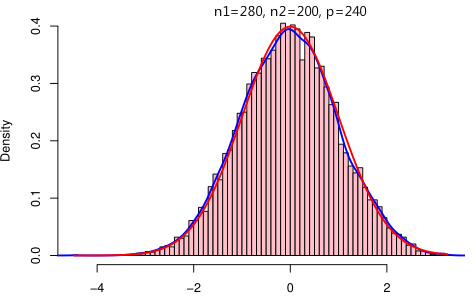}}
	\caption{
         The pictures show that for different $n_{1}, n_{2}, p$, that is when $y_{1}<1, y_{2}>1$, the probability density curve of the standard normal distribution N(0,1) is consistent with  the histogram of the proposed ML test as  $p$  becomes larger.}
	\label{fig3}
\end{figure}

\begin{figure}[htbp]
	\centering
	\subfigure{\includegraphics[height=3cm,width=3cm]{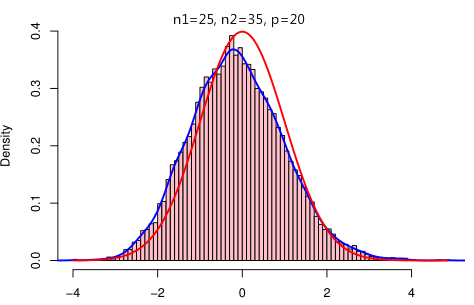}}
	\subfigure{\includegraphics[height=3cm,width=3cm]{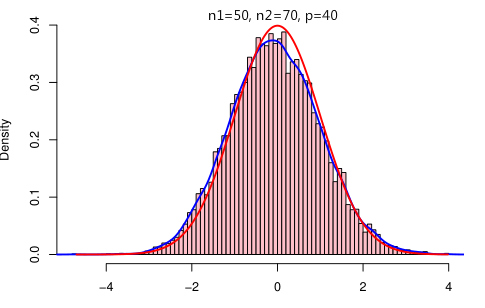}}
	\subfigure{\includegraphics[height=3cm,width=3cm]{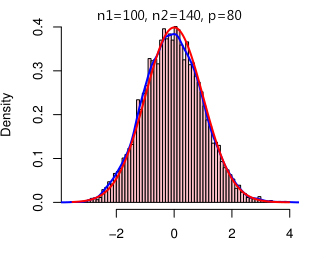}}
	\subfigure{\includegraphics[height=3cm,width=3cm]{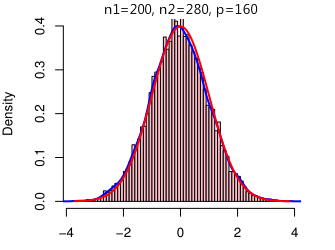}}
	\caption{
		   The pictures show that for different $n_{1}, n_{2}, p$, that is when $y_{1}<1, y_{2}<1$, the probability density curve of the standard normal distribution N(0,1) is consistent with  the histogram of the proposed ML test as  $p$  becomes larger .}
	\label{fig4}
\end{figure}

\begin{figure}[htbp]
	\centering
	\subfigure{\includegraphics[width=0.35\textwidth]{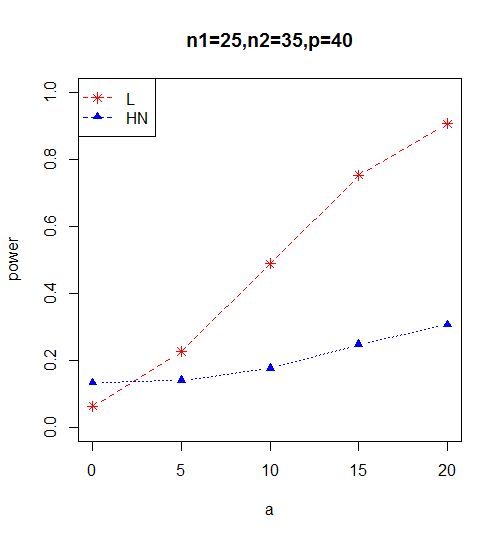}}
	\subfigure{\includegraphics[width=0.35\textwidth]{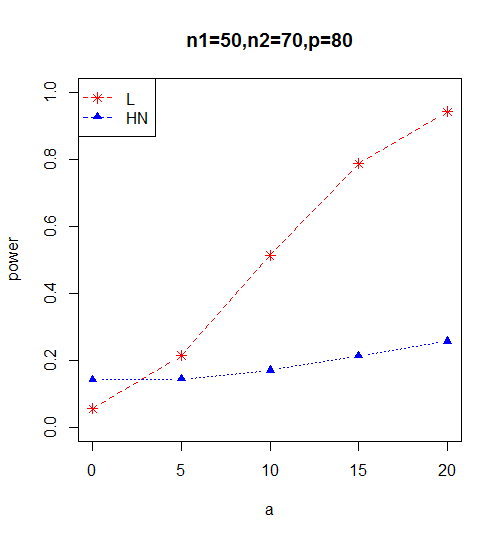}}
	\subfigure{\includegraphics[width=0.35\textwidth]{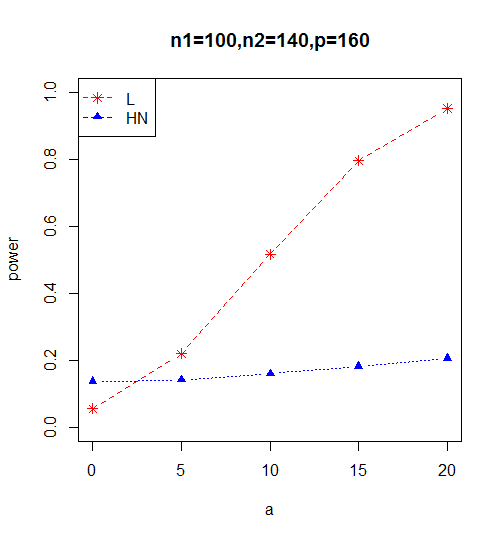}}
	\subfigure{\includegraphics[width=0.35\textwidth]{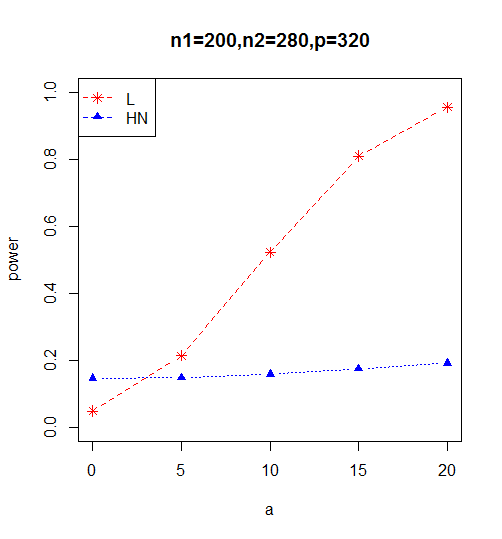}}
	\caption{
		 The  empirical powers are estimated based on 10,000 replications with real Gamma variables  and these results are based on the significance level of $\alpha=0.05$
		under Model \Rmnum{1} (note that our proposed test is abbreviated to L in the picture).}
	\label{fig5}
\end{figure}

\begin{figure}[htbp]
	\centering
	\subfigure{\includegraphics[width=0.35\textwidth]{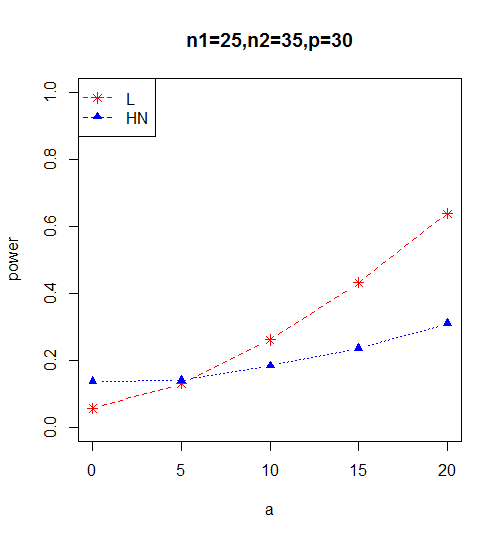}}
	\subfigure{\includegraphics[width=0.35\textwidth]{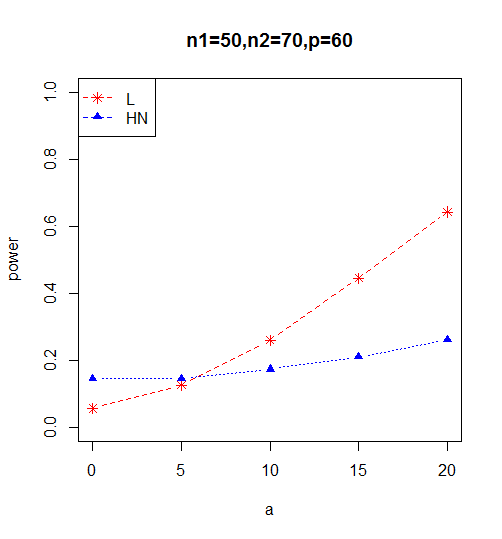}}
	\subfigure{\includegraphics[width=0.35\textwidth]{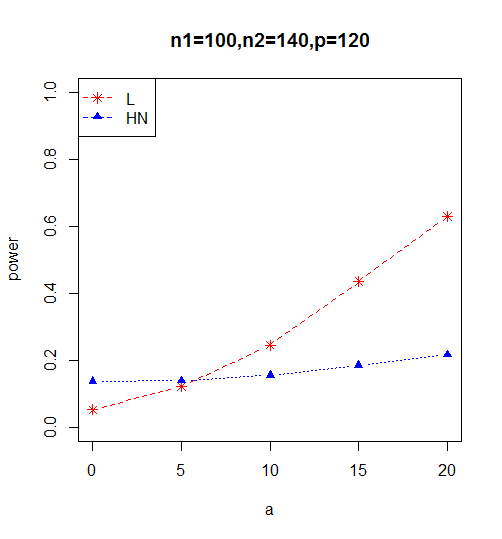}}
	\subfigure{\includegraphics[width=0.35\textwidth]{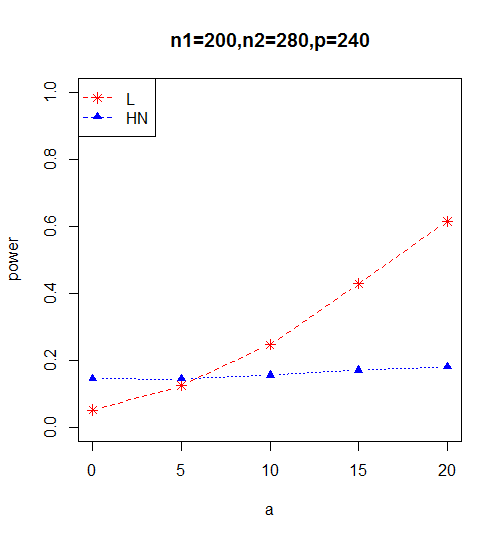}}
	\caption{
	 The  empirical powers are estimated based on 10,000 replications with real Gamma variables  and these results are based on the significance level of $\alpha=0.05$ under Model \Rmnum{1} (note that our proposed test is abbreviated to L in the picture).}
	\label{fig6}
\end{figure}

\begin{figure}[htbp]
	\centering
	\subfigure{\includegraphics[width=0.35\textwidth]{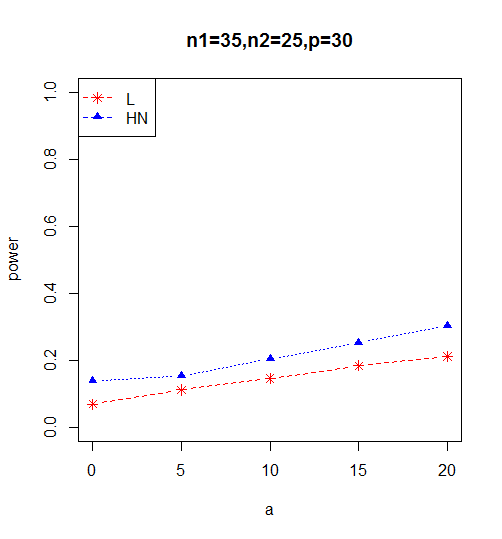}}
	\subfigure{\includegraphics[width=0.35\textwidth]{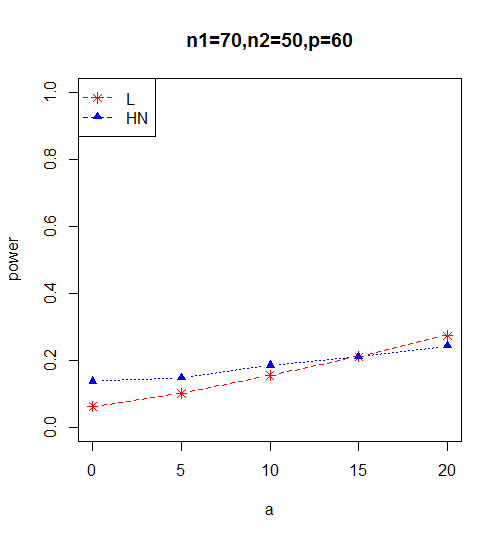}}
	\subfigure{\includegraphics[width=0.35\textwidth]{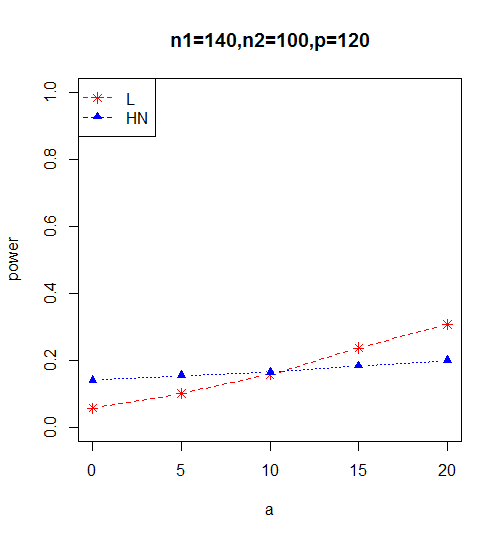}}
	\subfigure{\includegraphics[width=0.35\textwidth]{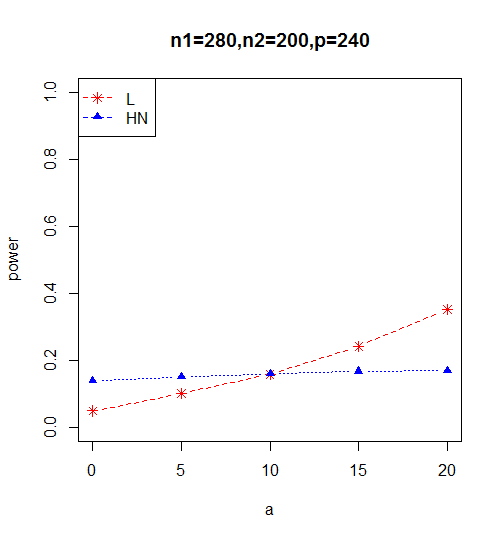}}
	\caption{
	 The empirical powers are estimated based on 10,000 replications with real Gamma variables  and these results are based on the significance level of $\alpha=0.05$ under Model \Rmnum{1} (note that our proposed test is abbreviated to L in the picture).}
	\label{fig7}
\end{figure}

\begin{figure}[htbp]
	\centering
	\subfigure{\includegraphics[width=0.47\textwidth]{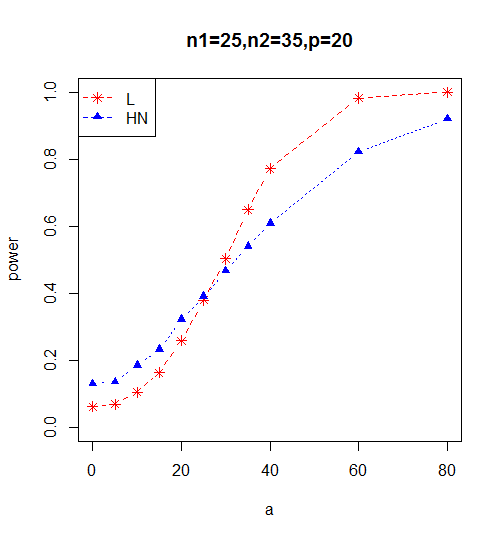}}
	\subfigure{\includegraphics[width=0.47\textwidth]{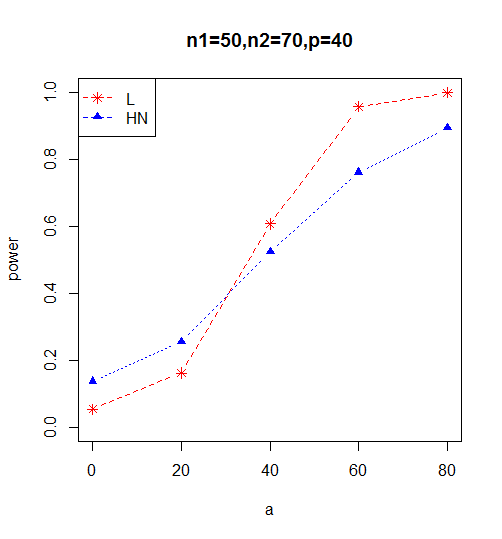}}
	\subfigure{\includegraphics[width=0.47\textwidth]{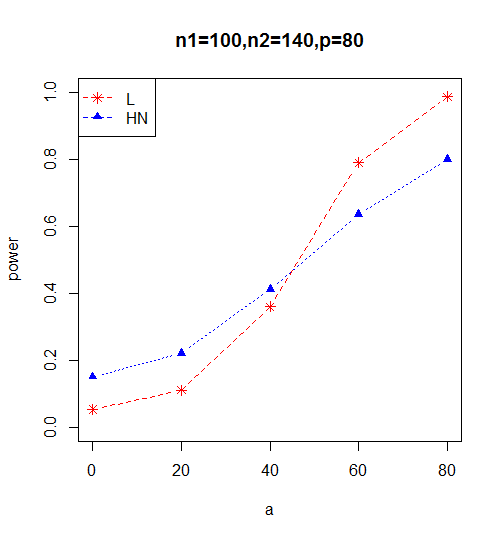}}
	\subfigure{\includegraphics[width=0.47\textwidth]{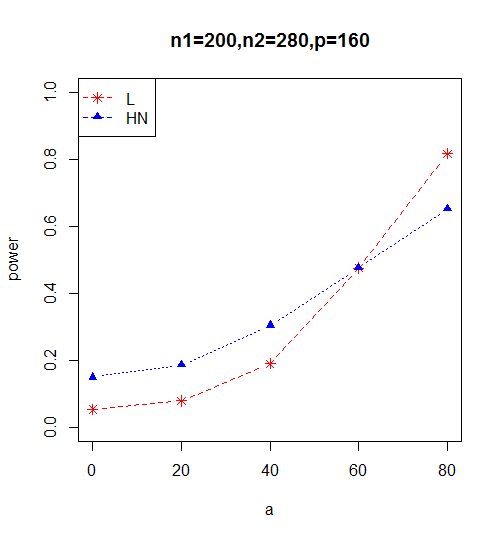}}
	\caption{
	 The empirical powers are estimated based on 10,000 replications with real Gamma variables  and these results are based on the significance level of $\alpha=0.05$ under Model \Rmnum{1} (note that our proposed test is abbreviated to L in the picture). }
	\label{fig8}
\end{figure}

\section{Conclusions and discussions}\label{sec5}

This paper is concerned with the new modified LRT for simultaneous testing of equalities of mean vectors and covariance matrices in high-dimensional settings. We prove that the modified LRT converges in distribution to the normal distribution under the null hypothesis. Simulation results show that the performance of the modified LRT is remarkable under applicable conditions, especially in Model \Rmnum{1} (we may refer to as the unbounded spectral norm of the population covariance matrix).

In this paper, we consider the unstructured covariance matrix. However, various models with specific covariance structures naturally appear in the context of repeated measurements with high-dimensional data, such as first-order autoregressive (AR),
moving average (MA), compound symmetry (CS), and so on. It remains to be a theoretical interest to study simultaneous testing of mean vectors and covariance matrices under such structured covariance models.
In addition, due to the lack of progress in RMT, we do not consider the asymptotic distribution of modified LRT under the alternative hypothesis in our current work. These will be continued for our future research.

\backmatter




\bmhead{Funding}
Niu's research was supported by the Natural Science Foundation of Shandong Province (No. ZR2021QA077);
Luo's research was supported by Statistical Research Project of Zhejiang Province   (No.  22TJJD05);
Bai's research was supported by National Natural Science Foundation of China (No. 12271536), National Natural Science Foundation of China (No. 12171198), and Natural Science Foundation of Jilin Province (No. 20210101147JC).

\section*{Declarations}
\bmhead{ Conflictof interest} On behalf of all authors, the corresponding author states that there is no conflict of interest.

\begin{appendices}

\section{Proofs of the main results}\label{sec4}

In this section, we provide the proof of Theorem \ref{thm}. We begin by enumerating key results  from RMT, which will be useful for our proof.
\begin{lem}\label{lem3}
	For any positive defined matrix  $\mathbf{W} \in \mathbb{R}^{n\times n}$ and $ \mathbf{q}\in\mathbb{R}^{n}$, $\mathbf{W}+\mathbf{q}\mathbf{q}'$ are invertible, we have
	\begin{align*}
	&(\mathbf{W}+\mathbf{q}\mathbf{q}^{\prime})^{-1}=\mathbf{W}^{-1}-\frac{\mathbf{W}^{-1}\mathbf{q}\mathbf{q}'\mathbf{W}^{-1}}{1+\mathbf{q}'\mathbf{W}^{-1}\mathbf{q}},\\
	&\mathbf{q}'(\mathbf{W}+\mathbf{q}\mathbf{q}')^{-1}=\frac{\mathbf{q}'\mathbf{W}^{-1}}{1+\mathbf{q}'\mathbf{W}^{-1}\mathbf{q}}.
	\end{align*}

\end{lem}

\begin{lem}[ Theorem 2.1 in \cite{zhang2019invariant}]\label{lem1}
	To test the hypothesis
	$\mathrm{H}_{0a}:\boldsymbol{\Sigma}_{1}=\boldsymbol{\Sigma}_{2}  \hspace{1ex} \mathrm{v.s.}  \hspace{1ex} \mathrm{H}_{1a}:\boldsymbol{\Sigma}_{1} \neq \boldsymbol{\Sigma}_{2}$,
	assuming that the conditions of  Theorem \ref{thm} hold.
	Then  under  the null hypothesis $\mathbf{H}_{0a}$,
	\begin{eqnarray*}
		\frac{ZH-p\cdot l_{n}-\mu_{n}}{\nu_{n}} \xrightarrow{D} N(0,1),
	\end{eqnarray*}
	where
	\begin{eqnarray*}\label{ZH}
		ZH=\sum\limits_{\lambda_{i}^{\mathbf{B}_{n}}\in(0,1)}[c_{1}\log\lambda_{i}^{\mathbf{B}_{n}} + c_{2} \log(1-\lambda_{i}^{\mathbf{B}_{n}})],
	\end{eqnarray*}
	$l_{n}$, $\mu_{n}$ and $\nu_{n}$ are defined in Theorem \ref{thm}.
	
\end{lem}

\begin{lem}[Lemma 5.3 in  \cite{zheng2015substitution}]\label{lem2}
	For $t=1,2$, let
	\begin{eqnarray*}
		&&\mathbf{\Delta}^{(t)}=\frac{1}{n_{t}N_{t}}\sum_{j\neq k \in \mathcal{U} }\mathbf{z}_{j}^{(t)}(\mathbf{z}_{k}^{(t)})',
		\quad  \mathbf{B}^{\mathbf{z}}_{t}=\frac{1}{N_{t}}\sum_{j=1}^{N_{t}}\mathbf{z}_{j}^{(t)}(\mathbf{z}_{j}^{(t)})',
		\quad  \mathbf{A}_{t}(z)=\mathbf{B}_{t}^{\mathbf{z}}-z\mathbf{I}_{p}.
	\end{eqnarray*}
	After truncation and normalization, we have
	\begin{eqnarray*}
		\mathbb{E}|\mathrm{tr}(\mathbf{A}_{t}^{-1}(z)\mathbf{\Delta}^{(t)})|^{2}=\mathbb{E}|\frac{1}{n_{t}N_{t}}\sum_{j\neq k \in \mathcal{U}}(\mathbf{z}_{j}^{(t)})'\mathbf{A}_{t}^{-1}(z)\mathbf{z}_{k}^{(t)}|^{2} \leq K N_{t}^{-1}
	\end{eqnarray*}
	for every $z\in \mathbb{C}^{+}$. Especially for every $z\in \mathbb{C}^{+}$,
	\begin{eqnarray*}
		\mathbb{E}|\mathrm{tr}(\mathbf{A}_{t}^{-2}(z)\mathbf{\Delta}^{(t)})|^{2}=\mathbb{E}|\frac{1}{n_{t}N_{t}}\sum_{j\neq k \in \mathcal{U}}(\mathbf{z}_{j}^{(t)})'\mathbf{A}_{t}^{-2}(z)\mathbf{z}_{k}^{(t)}|^{2}=O(N_{t}^{-1}), \quad \mathcal{U}=\{ 1, 2, \ldots, N_{t} \}.
	\end{eqnarray*}
\end{lem}

\begin{proof}[Proof of Theorem \ref{thm}]
	Recall  the  modified  LRT  statistic
	\begin{align*}
	L=&\sum_{\lambda_{i}^{\mathbf{B}_{n}}\in(0,1)}[c_{1}\log\lambda_{i}^{\mathbf{B}_{n}} + c_{2} \log(1-\lambda_{i}^{\mathbf{B}_{n}})]-\log(1+\mathrm{T}_{n})\\
	=& ZH-\log(1+\mathrm{T}_{n}).
	\end{align*}
	It is easy to verify that the proof of Theorem \ref{thm} can be divided into two steps, one is to prove that $ZH$ converges in distribution to a normal distribution, and the other is to prove that $\mathrm{T}_{n}$ converges in probability to a constant.
	Applying  Lemma \ref{lem1},  we can obtain that
	\begin{align*}
	\frac{ZH-p\cdot l_{n}-\mu_{n}}{\nu_{n}} \xrightarrow{D} N(0,1),
	\end{align*}
	where $ l_{n}$, $\mu_{n}$ and $\nu_{n}$ are defined in Theorem \ref{thm}.
	Thus,  it remains to prove that
	\begin{align}\label{LimitTn}
	\mathrm{T}_{n}-\frac{r_{n}}{1-r_{n}} \xrightarrow{p} 0.
	\end{align}
	Write
	\begin{align*}
	\mathrm{T}_{n}=&\frac{n_{1}n_{2}}{n}(\bar{\mathbf{z}}^{(1)}-\bar{\mathbf{z}}^{(2)})^{'}(n_{1}\mathbf{S}_{1}+n_{2}\mathbf{S}_{2})^{-1}(\bar{\mathbf{z}}^{(1)}-\bar{\mathbf{z}}^{(2)})\\
	=&\frac{n_{1}n_{2}}{n^{2}}(\bar{\mathbf{z}}^{(1)}-\bar{\mathbf{z}}^{(2)})^{'}\mathbf{S}_{n}^{-1}(\bar{\mathbf{z}}^{(1)}-\bar{\mathbf{z}}^{(2)})\\
	=&\frac{n_{1}n_{2}}{n^{2}}[(\bar{\mathbf{z}}^{(1)})'\mathbf{S}_{n}^{-1}\bar{\mathbf{z}}^{(1)}-2(\bar{\mathbf{z}}^{(1)})'\mathbf{S}_{n}^{-1}\bar{\mathbf{z}}^{(2)}+(\bar{\mathbf{z}}^{(2)})'\mathbf{S}_{n}^{-1}\bar{\mathbf{z}}^{(2)}],
	\end{align*}
	where
	\begin{align*}
	\mathbf{S }_{n}= \frac{1}{n}(n_{1}\mathbf{S}_{1}+n_{2}\mathbf{S}_{2})
	= \frac{1}{n}\sum_{t=1}^{2}\sum_{i=1}^{N_{t}}(\mathbf{z}_{i}^{(t)}-\bar{\mathbf{z}}^{(t)})(\mathbf{z}_{i}^{(t)}-\bar{\mathbf{z}}^{(t)})'.
	\end{align*}

	\subsection{The limit of $(\bar{\mathbf{z}}^{(t)})^{'}\mathbf{S}_{n}^{-1}\bar{\mathbf{z}}^{(t)}$}\label{section}
	
	In this section, our aim is to prove that
	\begin{eqnarray}\label{limit}
	(\bar{\mathbf{z}}^{(t)})^{'}\mathbf{S}_{n}^{-1}\bar{\mathbf{z}}^{(t)}-\frac{n}{N_{t}}\frac{ r_{n}}{1-r_{n}} \xrightarrow{p} 0, \quad t=1,2.
	\end{eqnarray}
	Let
	\begin{eqnarray*}
		\mathbf{S}_{nt}=\mathbf{S}_{n}+\frac{N_{t}}{n}\bar{\mathbf{z}}^{(t)}(\bar{\mathbf{z}}^{(t)})', \quad
		\mathbf{S}_{it}=\mathbf{S}_{nt}-\frac{1}{n}\mathbf{z}_{i}^{(t)}(\mathbf{z}_{i}^{(t)})',
	\end{eqnarray*}
	for $i=1,\ldots,N_{t}$, $t=1,2$.
	Therefore, from  Lemma \ref{lem3},  we have
	\begin{eqnarray*}
		(\bar{\mathbf{z}}^{(t)})^{'}\mathbf{S}_{n}^{-1}\bar{\mathbf{z}}^{(t)}
		=\frac{(\bar{\mathbf{z}}^{(t)})^{'}\mathbf{S}_{nt}^{-1}\bar{\mathbf{z}}^{(t)}}{1-\frac{N_{t}}{n}(\bar{\mathbf{z}}^{(t)})^{'}\mathbf{S}_{nt}^{-1}\bar{\mathbf{z}}^{(t)}}.
	\end{eqnarray*}
	To prove \eqref{limit}, we only need to prove  that
	\begin{eqnarray*}
		(\bar{\mathbf{z}}^{(t)})^{'}\mathbf{S}_{nt}^{-1}\bar{\mathbf{z}}^{(t)}-\frac{n}{N_{t}}r_{n} \xrightarrow{p} 0.
	\end{eqnarray*}
	Applying $\bar{\mathbf{z}}^{(t)}=N_{t}^{-1}\sum_{i=1}^{N_{t}}\mathbf{z}_{i}^{(t)}$,  we obtain
	\begin{eqnarray*}
		(\bar{\mathbf{z}}^{(t)})^{'}\mathbf{S}_{nt}^{-1}\bar{\mathbf{z}}^{(t)}
		&=&(\frac{1}{N_{t}}\sum\limits_{i=1}^{N_{t}}\mathbf{z}_{i}^{(t)})^{'}\mathbf{S}_{nt}^{-1}(\frac{1}{N_{t}}\sum\limits_{j=1}^{N_{t}}\mathbf{z}_{j}^{(t)})\\
		&=&d_{n1}+d_{n2},
	\end{eqnarray*}
	where
	\begin{eqnarray*}
		d_{n1}=\frac{1}{N_{t}^{2}}\sum\limits_{i=1}^{N_{t}}(\mathbf{z}_{i}^{(t)})^{'}\mathbf{S}_{nt}^{-1}\mathbf{z}_{i}^{(t)}, \quad
		d_{n2}=\frac{1}{N_{t}^{2}}\sum\limits_{i\neq j}^{N_{t}}(\mathbf{z}_{i}^{(t)})^{'}\mathbf{S}_{nt}^{-1}\mathbf{z}_{j}^{(t)}.
	\end{eqnarray*}
	Subsequently, it becomes necessary to evaluate the limits   of $d_{n1}$ and $d_{n2}$.
	As for $d_{n1}$, from  Lemma \ref{lem3}, we conclude that
	\begin{eqnarray*}
		(\mathbf{z}_{i}^{(t)})^{'}\mathbf{S}_{nt}^{-1}\mathbf{z}_{i}^{(t)}
		=\frac{(\mathbf{z}_{i}^{(t)})^{'}\mathbf{S}_{it}^{-1}\mathbf{z}_{i}^{(t)}}{1+\frac{1}{n}(\mathbf{z}_{i}^{(t)})^{'}\mathbf{S}_{it}^{-1}\mathbf{z}_{i}^{(t)}}.
	\end{eqnarray*}
	In addition,  using the fact that by Theorem 2.1 and Theorem 2.3  in \cite{ha2022ridgelized},  for fixed $\kappa>0$,
	\begin{eqnarray*}
		\frac{1}{n}(\mathbf{z}_{i}^{(t)})^{'}(\mathbf{S}_{it}+\kappa\textbf{I}_{p})^{-1}\mathbf{z}_{i}^{(t)}-\frac{p}{n}\cdot \Theta(r,\kappa) \xrightarrow{p}0,
	\end{eqnarray*}
	where
	\begin{eqnarray*}
		\Theta(r,\kappa)=\frac{1-\kappa m_{F}(-\kappa)}{1-r(1-\kappa m_{F}(-\kappa))},
	\end{eqnarray*}
	$m_{F}(-\kappa)$ is a substitute for $m_{F}(z)$  by replacing parameter $z$ with  $-\kappa$,  and $m_{F}(z)$ is the Stieltjes transform of the limiting spectral distribution (LSD) of $\mathbf{S}_{it}$. \cite{chen2011regularized} have used  $\hat{\Theta}_{n}(\kappa)$ to estimate $\Theta(r,\kappa)$ under normal distribution,
	where $\hat{\Theta}_{n}(\kappa)$  is  substituted for  $\Theta(r,\kappa)$  with the parameters $r$ and $m_{F}(-\kappa)$ replaced by $r_{n}$ and $m_{F_{n}}(-\kappa)$ respectively,
	and  $m_{F_{n}}(-\kappa)=p^{-1}\text{tr}(\mathbf{S}_{it}+\kappa\textbf{I}_{p})^{-1}$.
	Finally, using  Vitali's Convergence Theorem, we have
	\begin{align*}
	\frac{1}{n}(\mathbf{z}_{i}^{(t)})^{'}\mathbf{S}_{it}^{-1}\mathbf{z}_{i}^{(t)} - \frac{r_{n}}{1-r_{n}} \xrightarrow{p} 0.
	\end{align*}
	Note that  because of $(\mathbf{z}_{i}^{(t)})^{'}\mathbf{S}_{nt}^{-1}\mathbf{z}_{i}^{(t)}$  has the same distribution for every $i=1,\ldots, N_{t}$,
	thus
	\begin{align*}
	d_{n1}-\frac{n}{N_{t}}r_{n} \xrightarrow{p} 0.
	\end{align*}
	Now we consider the second term $d_{n2}$. According to Lemma \ref{lem2}, it is clear that
	\begin{align*}
	\mathbb{E}|\frac{1}{N_{t}^{2}}\sum\limits_{i\neq j}^{N_{t}}(\mathbf{z}_{i}^{(t)})^{'}(\mathbf{S}_{nt}+\kappa\textbf{I}_{p})^{-1}\mathbf{z}_{j}^{(t)}|^2\leq K N_{t}^{-1},
	\end{align*}
	and using  Vitali's  Convergence Theorem again,
	we obtain
	\begin{align*}
	d_{n2} \xrightarrow{p} 0.
	\end{align*}
	Thus, the proof of \eqref{limit} is  complete.

	\subsection{The limit of $(\bar{\mathbf{z}}^{(1)})'\mathbf{S}_{n}^{-1}\bar{\mathbf{z}}^{(2)}$}
	
	In this section, our goal is to show that
	\begin{eqnarray}\label{limit4}
	(\bar{\mathbf{z}}^{(1)})'\mathbf{S}_{n}^{-1}\bar{\mathbf{z}}^{(2)}\xrightarrow{p} 0.
	\end{eqnarray}
	Let
	$\mathbf{S}_{n21}=\mathbf{S}_{n2}+n^{-1}N_{1}\bar{\mathbf{z}}^{(1)}(\bar{\mathbf{z}}^{(1)})'$
	and recalling the definition of $\mathbf{S}_{n2}=\mathbf{S}_{n}+n^{-1}N_{2}\bar{\mathbf{z}}^{(2)}(\bar{\mathbf{z}}^{(2)})'$,
	using Lemma \ref{lem3} again, we obtain
	\begin{align*}
	(\bar{\mathbf{z}}^{(1)})'\mathbf{S}_{n}^{-1}\bar{\mathbf{z}}^{(2)}
	=\frac{(\bar{\mathbf{z}}^{(1)})'\mathbf{S}_{n21}^{-1}\bar{\mathbf{z}}^{(2)}}
	{(1-\frac{N_{2}}{n}(\bar{\mathbf{z}}^{(2)})'\mathbf{S}_{n2}^{-1}\bar{\mathbf{z}}^{(2)})(1-\frac{N_{1}}{n}(\bar{\mathbf{z}}^{(1)})'\mathbf{S}_{n21}^{-1}\bar{\mathbf{z}}^{(1)})}.
	\end{align*}
	In the previous section, we have obtained the limit of $(\bar{\mathbf{z}}^{(2)})'\mathbf{S}_{n2}^{-1}\bar{\mathbf{z}}^{(2)}$.
	Similarly, the limit of $(\bar{\mathbf{z}}^{(1)})'\mathbf{S}_{n21}^{-1}\bar{\mathbf{z}}^{(1)}$ can also be obtained,
	thus we need only consider the limit of $(\bar{\mathbf{z}}^{(1)})'\mathbf{S}_{n21}^{-1}\bar{\mathbf{z}}^{(2)}$.
	As in section \ref{section}, we then write
	\begin{align*}
	&(\bar{\mathbf{z}}^{(1)})'(\mathbf{S}_{n21}+\kappa\textbf{I}_{p})^{-1}\bar{\mathbf{z}}^{(2)} \\
	=& (\frac{1}{N_{1}}\sum\limits_{i=1}^{N_{1}}\mathbf{z}^{(1)}_{i})'(\mathbf{S}_{n21}+\kappa\textbf{I}_{p})^{-1}(\frac{1}{N_{2}}\sum\limits_{j=1}^{N_{2}}\mathbf{z}^{(2)}_{j})
	\\
	=& \frac{1}{N_{1}}\frac{1}{N_{2}}\sum\limits_{i=1}^{N_{1}}\sum\limits_{j=1}^{N_{2}}(\mathbf{z}^{(1)}_{i})'(\mathbf{S}_{n21}+\kappa\textbf{I}_{p})^{-1}\mathbf{z}^{(2)}_{j}.
	\end{align*}
	This together with Lemma \ref{lem2},  shows that
	\begin{align*}
	\mathbb{E}|\frac{1}{N_{1}}\frac{1}{N_{2}}\sum\limits_{i=1}^{N_{1}}\sum\limits_{j=1}^{N_{2}}(\mathbf{z}^{(1)}_{i})'(\mathbf{S}_{n21}
	+\kappa\textbf{I}_{p})^{-1}\mathbf{z}^{(2)}_{j}|^{2} \leq K(\sqrt{N_{1}N_{2}})^{-1}.
	\end{align*}
	Then using Vitali's Convergence Theorem, it is easily seen that
	\begin{align*}
	(\bar{\mathbf{z}}^{(1)})'\mathbf{S}_{n21}^{-1}\bar{\mathbf{z}}^{(2)}\xrightarrow{p} 0,
	\end{align*}
	we have completed the proof of \eqref{limit4}.
	
	Finally  summarizing the above, the proof of \eqref{LimitTn} is  complete.
	
\end{proof}






\end{appendices}

\bibliographystyle{plain}
\bibliography{references}

\end{document}